\documentclass[runningheads,a4paper]{llncs}

\usepackage{times}
\usepackage{amssymb}
\usepackage{graphicx}
\usepackage{url}
\usepackage{multirow}
\usepackage{amsmath}
\usepackage{bbm}

\newcommand{\keywords}[1]{\par\addvspace\baselineskip
\noindent\keywordname\enspace\ignorespaces#1}

\newcommand{\tabitem}{~~\llap{\textbullet}~~}

 \urldef{\mailsb}\path|{jlmk,swarupak,sparta,mouchta,kunzman}@amazon.com|

\begin{document}

\title{Exploiting Large-scale Teacher-Student Training for On-device Acoustic Models}

\titlerunning{Exploiting Large-scale Teacher-Student Training for On-device Acoustic Models}


\author{Jing Liu \inst{1} \and Rupak Vignesh Swaminathan \inst{1}  \and Sree Hari Krishnan Parthasarathi \inst{1} \and \\ 
Chunchuan Lyu \inst{2} \and Athanasios Mouchtaris \inst{1} \and Siegfried Kunzmann \inst{1}}

\authorrunning{Jing Liu et al.}

\institute{Alexa Machine Learning, Amazon, USA \\
 \and
 School of Informatics, University of Edinburgh, Scotland, UK \\
 \mailsb\\
}


\index{Liu, Jing}
\index{Swaminathan, Rupak}
\index{Parthasarathi, Hari}
\index{Lyu, Chunchuan}
\index{Mouchtaris, Athanasios}
\index{Kunzmann, Siegfried}
\toctitle{} \tocauthor{}

\maketitle

%
%
%
%
\begin{abstract}
We present results from Alexa speech teams on semi-supervised learning (SSL) of acoustic models (AM) with experiments spanning over 3000 hours of GPU time, making our study one of the largest of its kind. We discuss SSL for AMs in a small footprint setting, showing that a smaller capacity model trained with 1 million hours of unsupervised data can outperform a baseline supervised system by 14.3\% word error rate reduction (WERR). When increasing the supervised data to seven-fold, our gains diminish to 7.1\% WERR; to improve SSL efficiency at larger supervised data regimes, we employ a step-wise distillation into a smaller model, obtaining a WERR of 14.4\%. We then switch to SSL using larger student models in low data regimes; while learning efficiency with unsupervised data is higher, student models may outperform teacher models in such a setting. We develop a theoretical sketch to explain this behavior.

\end{abstract}

\section{Introduction}
\label{sec:intro}

Semi-supervised learning (SSL) has a rich history in automatic speech recognition (ASR)~\cite{deepspeech2:2016,gong-semi:2016,kemp1999unsupervised,lamel2002lightly,ma2006unsupervised}. Self-training is a commonly used technique employing confidence measures~\cite{siu1997improved,huang2013semi}. Student-teacher distillation techniques~\cite{ba2014deep,hinton2015distilling}, foregoing full decoders and confidence models, have been shown to be effective for SSL~\cite{li2014learning}. SSL methods for end-to-end ASR have been studied in~\cite{ibm1,Kahn_2020,weninger2020semisupervised,chen2020semisupervised,mun2019sequence}. Furthermore, investigations with SSL in combination with data augmentation, pretraining and iterative self-training are done in~\cite{zhang2020pushing} and~\cite{xie2020unsupervised}. 

Recently, student-teacher distillation techniques for hybrid HMM-LSTM models have been shown to scale to very large data sets (1 million hours) for models with high capacity~\cite{ssl_1mhr_paper,parthasarathi2019realizing}. The efficacy of model compression using student-teacher distillation is well established~\cite{mirzadeh2019improved,KD_Waters_2016,Watanabe_TS_2017}. In this context we study learning curves for AM for two tasks: (a) smaller footprint modeling, and (b) low training-data regimes. 

Our motivation for low footprint AM comes from edge computing, where models are capacity restricted in terms of compute and memory \cite{shi2016edge,yu2017survey,he2018streaming}. We are interested in understanding if SSL, at very large data regimes for small models, can still yield gains in accuracy. In~\cite{Moriya2018}, mean squared errors between the teacher and the student hidden representations are explored as a regularization term in knowledge distillation. We demonstrate that a step-wise distillation approach, introduced in~\cite{mirzadeh2019improved} can be effective, although this comes at the cost of more computation at training time. In low data regimes, SSL is an effective technique to reduce annotation costs \cite{kemp1999unsupervised,lamel2002lightly,ma2006unsupervised,Chen2019}. For our second task, using knowledge distillation for SSL, we find that to achieve a performance comparable to that of a fully supervised system, the proportion of required supervised data decreases as the amount of total data increases. However, we find that in low data regimes, students can be better than teachers. 

\begin{figure}[tp]
  \centering
  \centerline{\includegraphics[width=0.8\linewidth]{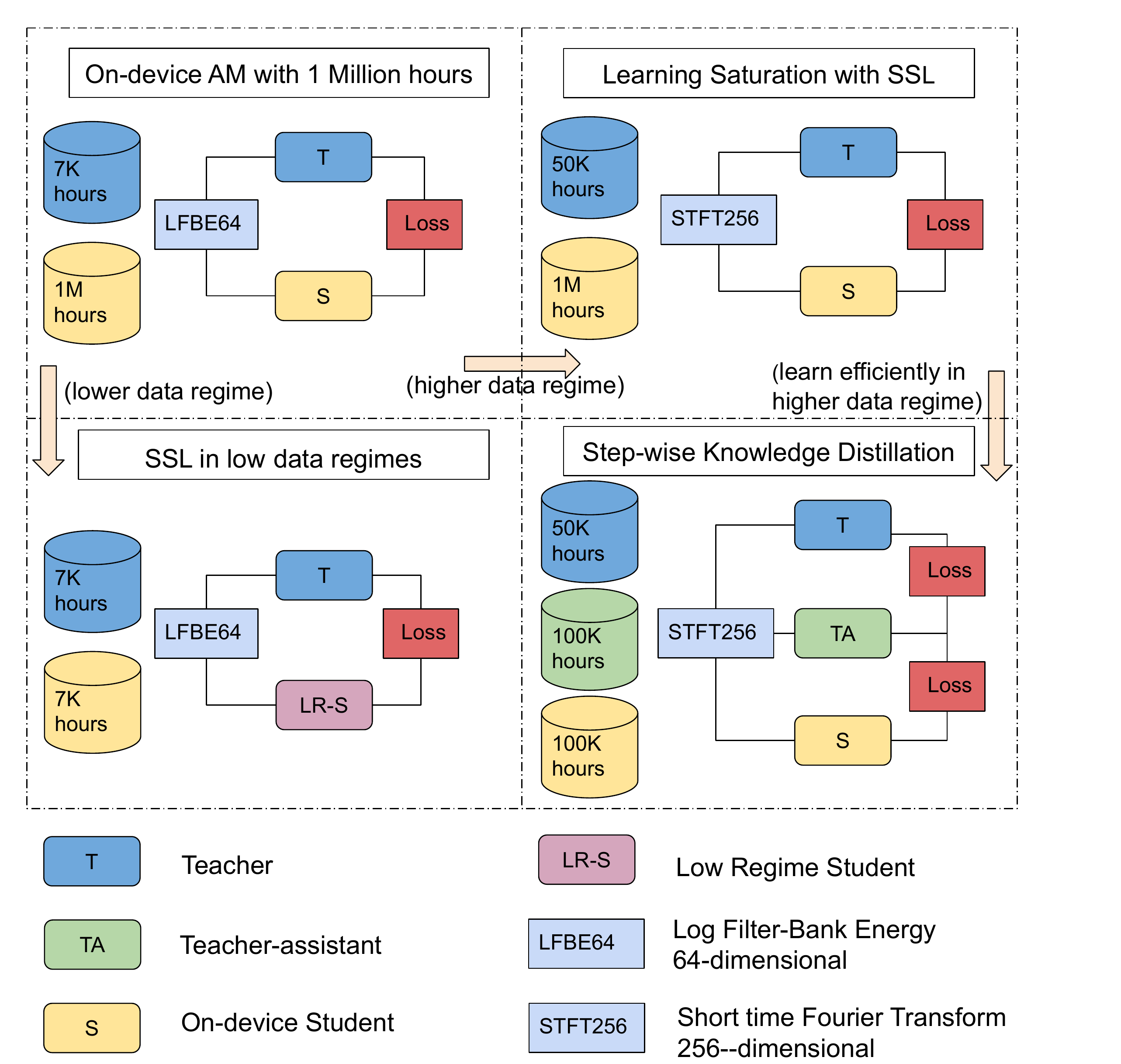}}
  \caption{\textit{Overview of our SSL approach progressing through 4 experiments. a) Top-left: Analysis of models using large unsupervised data. b) Top-right: Analysis of models in large supervised and unsupervised data regimes. c) Bottom-right: Improving learning efficiency in higher data regimes. d) Bottom-left: SSL studies in low data regimes. }}\medskip
\label{fig:overview}
\end{figure}

Our contributions in this work are as follows: (1) we establish the robustness of small capacity semi-supervised models trained on 1 million hours of data; (2) we show an effective way to mitigate the learning saturation problem at higher data regimes for an on-device acoustic model; (3) we report results of a model distilled from a teacher trained on transcribed low-resource data, and present an empirical risk analysis. Figure~\ref{fig:overview} describes the approach we follow in this paper. We begin with a discussion on model configurations in Section~\ref{sec:modeling}. In Section~\ref{sec:small_footprint_exp}, we present SSL for small footprint AMs. We show that such models can exploit a large amount of unsupervised data. However, when the amount of supervised data is increased from 7,000 hours to 52,000 hours the gains decrease. To mitigate this, we discuss a step-wise distillation into a smaller model. We then transition into experiments conducted in low resource settings in Section~\ref{sec:low_resource_exp}.

\section{System Description}
\label{sec:modeling}
We now describe acoustic model configurations used in this work. It is a hybrid system, with an LSTM~\cite{lstm_1997} estimating the senone posterior probabilities corresponding to clustered triphone HMM states. The HMMs are single state models using low-frame rate features~\cite{pundak2016lower}, which are computed on speech signals every 10 ms, with a 25 ms analysis window. A running causal mean estimate is computed and subtracted from the features, and the resulting features are normalized by applying a global mean and variance normalization. The models are trained with the cross-entropy (CE) criterion, followed by sequence discriminative training using state-level minimum Bayes risk (sMBR) loss~\cite{sMBR_2009}. We follow an exponential learning rate decay for twelve epochs. More details on model configurations can be found in Table~\ref{ssl_model_configs}.

\begin{table}[h!]
\caption{\textit{Model Configurations}}
\label{ssl_model_configs}
\vspace{1mm}
\begin{tabular}{l}
\hline
\textbf{Common configurations for all experiments}  \\ \hline
   \tabitem Features are stacked and subsampled to 33 Hz \\
   \tabitem 5 layers and 768 neurons/layer BLSTM teachers with 78 million parameters \\     
    \tabitem Teacher model is cross-entropy and sequence trained \\    
    \tabitem Teacher model trained with distributed trainer from~\cite{strom2015scalable}\\
    \tabitem Output posterior distribution over 3183 senones \\ 
    \tabitem Data selection for unsupervised data as in~\cite{ssl_1mhr_paper}  \\ 
    \tabitem Training/test data are sampled from de-identified speech data \\ 
    \tabitem Student training strategy with compressed posteriors as in~\cite{ssl_1mhr_paper} \\ 
    \tabitem TST1:  Test data consists of 100 hours of speech data \\ 
    \tabitem TST2:  Test data consists of 30 hours of speech data \\ \hline
\textbf{Section~\ref{subsec:On_device_AM_1M_speech} Configurations for on-device AM}  \\ \hline
   \tabitem 64-dimensional log filter-bank energies (LFBE) \\
   \tabitem 5 layers and 428 neurons/layer Uni-LSTM on-device student AMs with 8 million parameters  \\
    \tabitem Cross entropy training of student models with distributed trainer from~\cite{bmuf-2016} \\
    \tabitem sMBR training of student model with distributed trainer from~\cite{strom2015scalable} \\
   \tabitem 7,000 hours  of supervised training data  with TST1 test data\\ 
   \tabitem 1 million hours of unsupervised speech data       \\ \hline
\textbf{Section~\ref{subsec:On_device_SSL_saturation} Configurations for studying learning saturation} \\ \hline
   \tabitem 256-dimensional STFT features~\cite{Sak2015FastAA,flstm_2015} \\
     \tabitem Teacher and student model architectures: same as above \\ 
     \tabitem Distributed training strategy: same as above \\ 
   \tabitem 52 K hours supervised training data with TST1 test data\\
   \tabitem 1 million hours of unsupervised speech data                                  \\ \hline
\textbf{Section~\ref{sec:takd_exp} Configurations for step-wise distillation}    \\ \hline
   \tabitem 256-dimensional STFT features \\
     \tabitem Teacher and student model architectures: same as above \\ 
     \tabitem Distributed training strategy: same as above \\ 
   \tabitem 50 K hours supervised training data with TST1 test data\\
   \tabitem 100 K hours unsupervised data                             \\ \hline
\textbf{Section~\ref{sec:low_resource_exp} Configurations for low resource experiments}                          \\ \hline
   \tabitem 64-dimensional LFBE features \\
     \tabitem 5 layers and 768 neurons/layer Uni-LSTM students with 24 million parameters \\   
    \tabitem Teacher and student models trained with distributed trainer from~\cite{strom2015scalable}\\
   \tabitem Total training data is 7000 hours with TST2 test data\\
    \tabitem Supervised data ranges from 100 hours up to 7000 hours \\
    \tabitem The rest of the data is treated as unsupervised data\\
                                                              
\end{tabular}
\end{table}

\section{SSL for small footprint AM}
\label{sec:small_footprint_exp}
In this section, we begin with an analysis of acoustic model complexity with supervised and unsupervised data for smaller and larger model footprints. We then investigate a method for efficient model distillation in larger supervised data regimes.

\subsection{Learning Curves on Large Unsupervised Data}
\label{subsec:On_device_AM_1M_speech}
In Figure \ref{fig:wer_lfbe_edge_ssl_vs_cloud}, on TST1 test data, we analyze the learning curves for on-device and cloud student models, taught by the same teacher network. Accuracy is reported as relative word error rate reduction (WERR)~\cite{parthasarathi2015fmllr,garimella2015robust}. Given model A's WER ($\text{WER}_A$) and a baseline B's WER ($\text{WER}_B$), the WERR of A over B is computed as
$
\text{WERR} = (\text{WER}_B - \text{WER}_A)/\text{WER}_B.
$
\begin{figure}[th]
  \centering
  \centerline{\includegraphics[width=0.8\linewidth]{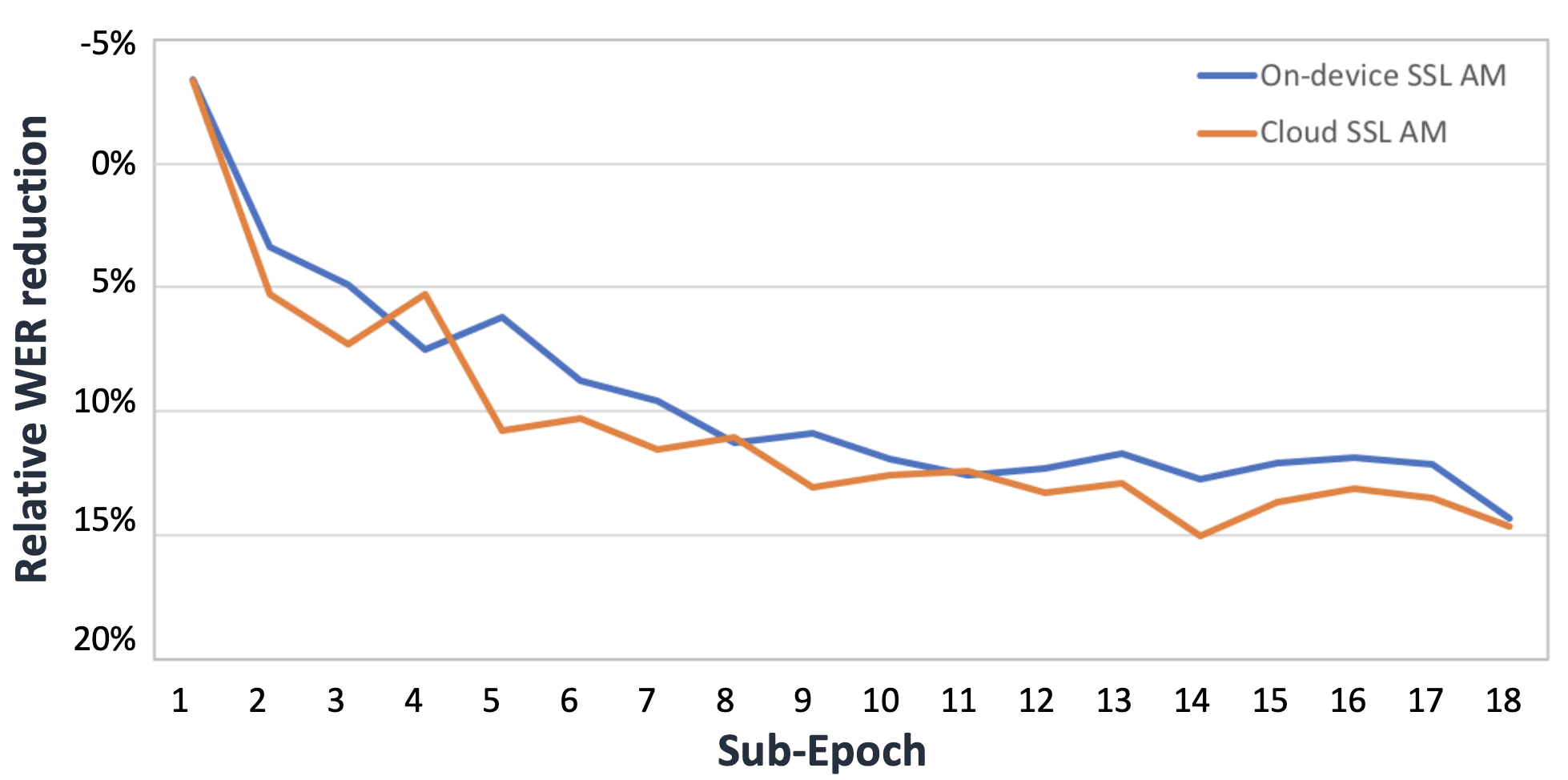}}
  \vspace{0.2cm}
  \caption{\textit{On TST1: relative WERR (\%) per sub-epoch of the 1 million hour SSL models against baseline LSTM AMs that are trained with CE criterion on the fully supervised 7,000 hour training data.}}\medskip
\label{fig:wer_lfbe_edge_ssl_vs_cloud}
\end{figure}

The vertical axis is the relative WERR (\%) against baseline LSTM AMs which are trained with CE criterion on the fully supervised 7,000 hour training data. The horizontal axis corresponds to the amount of unsupervised data. Each sub-epoch in the axis corresponds to about 60,000 hours of unsupervised data, totaling to 1 million hours. From Figure \ref{fig:wer_lfbe_edge_ssl_vs_cloud}, we see that the relative WERR improves steadily for both cloud and on-device student AMs as we use increasing amount of unsupervised data. When using 1 million hours unsupervised data, the relative WERR for on-device and cloud student AMs are 14.3\% and 14.6\% respectively in CE comparing with its fully supervised models.

We perform sequence level discriminative training only on the 7,000 hour supervised dataset, demonstrating that the gains at the CE stage also carry over to the sMBR stage. Table~\ref{tab:lfbe_data_wer} shows WERR for on-device and cloud student AMs using the same 4-gram LM. The WERRs for SSL CE + sMBR and supervised CE + sMBR are 25.7\% and 18.2\% respectively for on-device AMs.

\begin{table}[th]
\caption{\textit{On TST1: relative WERR (\%) for sequence training of SSL students. The baseline LSTM AMs are trained with CE criterion in a fully supervised setting with 7,000 hours of data}}
\label{tab:lfbe_data_wer}
\vspace{-1mm}
\begin{center}
\begin{tabular}{|l|l|l|}
\hline
\multicolumn{3}{|c|}{\textbf{WERR (\%)}}                          \\ \hline
\textbf{System}            & \textbf{On-device} & \textbf{Cloud} \\ \hline
Baseline supervised CE        & 0                  & 0              \\ \hline
Baseline supervised CE + sMBR & 18.2               & 9.7            \\ \hline
SSL CE                     & 14.3               & 14.6           \\ \hline
SSL CE + sMBR              & 25.7               & 24.0           \\ \hline
\end{tabular}
\end{center}
\end{table}

\subsection{Learning Saturation for On-device SSL System}
\label{subsec:On_device_SSL_saturation}
In the following experiment, we increase the supervised data by about 7 times to 52,000 hours. We then gradually increase the amount of unsupervised data set to 1 million hours, computing WERR against a supervised model trained with 52,000 hours of supervised data. From Table \ref{tab:stft_data_wer}, we get a relative WERR of 6.06\% by adding up to 540 K hours of unsupervised data. If we further increase unsupervised data from 540 K hours to 1 million hours, the additional relative WERR is only 1\%. We conclude that learning starts to saturate at 540 K hours. Adding additional unsupervised data helps little with accuracy improvement.

 \vspace{-0.4cm}
\begin{table}[t]
\caption{\textit{On TST1: relative WERR (\%) for SSL student AM with 8M parameters. Both the baseline and SSL AM are trained with 52k hours of fully supervised data.}}
\label{tab:stft_data_wer}
 \vspace{-0.1cm}
\begin{center}
\begin{tabular}{|l|l|l|}
\hline
   \textbf{WERR(\%)}&  \textbf{Unsupervised (hours)} \\ \hline
  6.05\%      & 540 K\\ \hline
  6.62\%  &720 K \\ \hline
 6.86\%  &  900 K     \\ \hline
  7.05\%  &  1 M   \\ \hline
\end{tabular}
\end{center}
\end{table}

\subsection{Improving the Efficiency of Knowledge Distillation}
\label{sec:takd_exp}
In this section, we study a distillation method to improve the learning efficiency at higher data regimes. Specifically, we reduce the gap between the teacher and student models through an intermediate teacher assistant~\cite{mirzadeh2019improved}.

The supervised part of the training data consists of about 50 K hours of supervised US English speech, and the unsupervised part consists of about 100 K hours of unsupervised data. The unsupervised training data is used for training the teacher assistant. We perform step-wise teacher assisted knowledge distillation. The first-step knowledge distillation happens through a 78-million-parameter bidirectional LSTM teacher model and a 28-million-parameter teacher assistant model. The second-step knowledge distillation occurs between the 28-million-parameter teacher assistant model and an 8-million-parameter student network using only the soft-targets from the teacher assistant. 

Table \ref{tab:KD} shows relative WERR against a baseline cross-entropy student model (row 4) trained with 50,000 hours of supervised data. Rows 1, 2 and 3 show WERR of the teacher, supervised cloud model, and SSL-trained teacher assistant models evaluated against baseline. Sequence training with sMBR (in row 4) shows 17.7\% relative WERR in a fully supervised system. The distillation results are presented in rows 5 and 6. The improvement over direct Knowledge Distillation (KD) is 9.2\% relative. For Teacher Assisted Knowledge Distillation (TAKD), there is a substantial improvement of 14.4\% compared to baseline supervised system, thus showing that step-wise distillation with a teacher assistant indeed helps with the efficiency for smaller models in large data regimes. 

\begin{table}[tp]
\caption{\textit{On TST1: relative WERR (\%) for teacher assisted knowledge distillation method. The baseline model is the cross-entropy trained model on supervised data (row 4). WERR values are computed against baseline model's WER.}}
\label{tab:KD}
\vspace{2mm}
\begin{center}
\begin{tabular}{|c|c|c|c|c|}
\hline
\multicolumn{1}{|l|}{\textbf{Index}} & \multicolumn{1}{c|}{\textbf{AM}} & \multicolumn{1}{l|}{\textbf{Params (M)}} & \multicolumn{1}{l|}{\textbf{Stage}} & \multicolumn{1}{l|}{\textbf{WERR (\%)}} \\ \hline
1                                    & Teacher                          & 78                                       & sMBR                                & 33.0                                    \\ \hline
2                                    & Cloud                            & 28                                       & sMBR                                & 21.67                                   \\ \hline
3                                    & TA                               & 28                                       & sMBR                                & 25.97                                   \\ \hline
\multirow{2}{*}{4}                   & \multirow{2}{*}{On-device}       & \multirow{2}{*}{8}                       & CE                                  & 0                                       \\ \cline{4-5} 
                                     &                                  &                                          & sMBR                                & 17.7                                    \\ \hline
\multirow{2}{*}{5}                   & \multirow{2}{*}{KD}              & \multirow{2}{*}{8}                       & CE                                  & 9.2                                     \\ \cline{4-5} 
                                     &                                  &                                          & sMBR                                & 15.7                                    \\ \hline
\multirow{2}{*}{6}                   & \multirow{2}{*}{TAKD}            & \multirow{2}{*}{8}                       & CE                                  & 14.4                                    \\ \cline{4-5} 
                                     &                                  &                                          & sMBR                                & 19.0                                    \\ \hline
\end{tabular}
\end{center}
\end{table}

\section{SSL in Low Resource Settings}
\label{sec:low_resource_exp}
In this section, we present our results on low resource settings with SSL; specifically, now the student model is larger, but the overall amount of data is restricted to 7,000 hours.

\subsection{Accuracy Gains With Unsupervised Data}
\label{sec:exp:gains_unsupervised}
Figure~\ref{fig:unsupervised} shows learning curves for student models at fixed amount of supervised data. The actual amounts of supervised data are: 100, 250, 500, 1000, 3500 and 7000 hours. Triangular markers on each curve correspond to student models trained with increasing amounts of unsupervised data by fixing the amount of supervised data. 

Relative WERR is computed against a baseline model trained only on the supervised data corresponding to that curve of the same color. The circular solid dots in this plot show the relative WERR for the teacher models. These dots line up vertically against the triangles corresponding to the baselines models.

For each curve, the slope of the curve becomes less steep as the amount of unsupervised data increases, meaning diminishing returns for additional data. Note that the student models can outperform the corresponding teacher model, once the student model observes unsupervised data.

\begin{figure}[ht]
  \centering
  \centerline{\includegraphics[width=0.60\linewidth, angle =270]{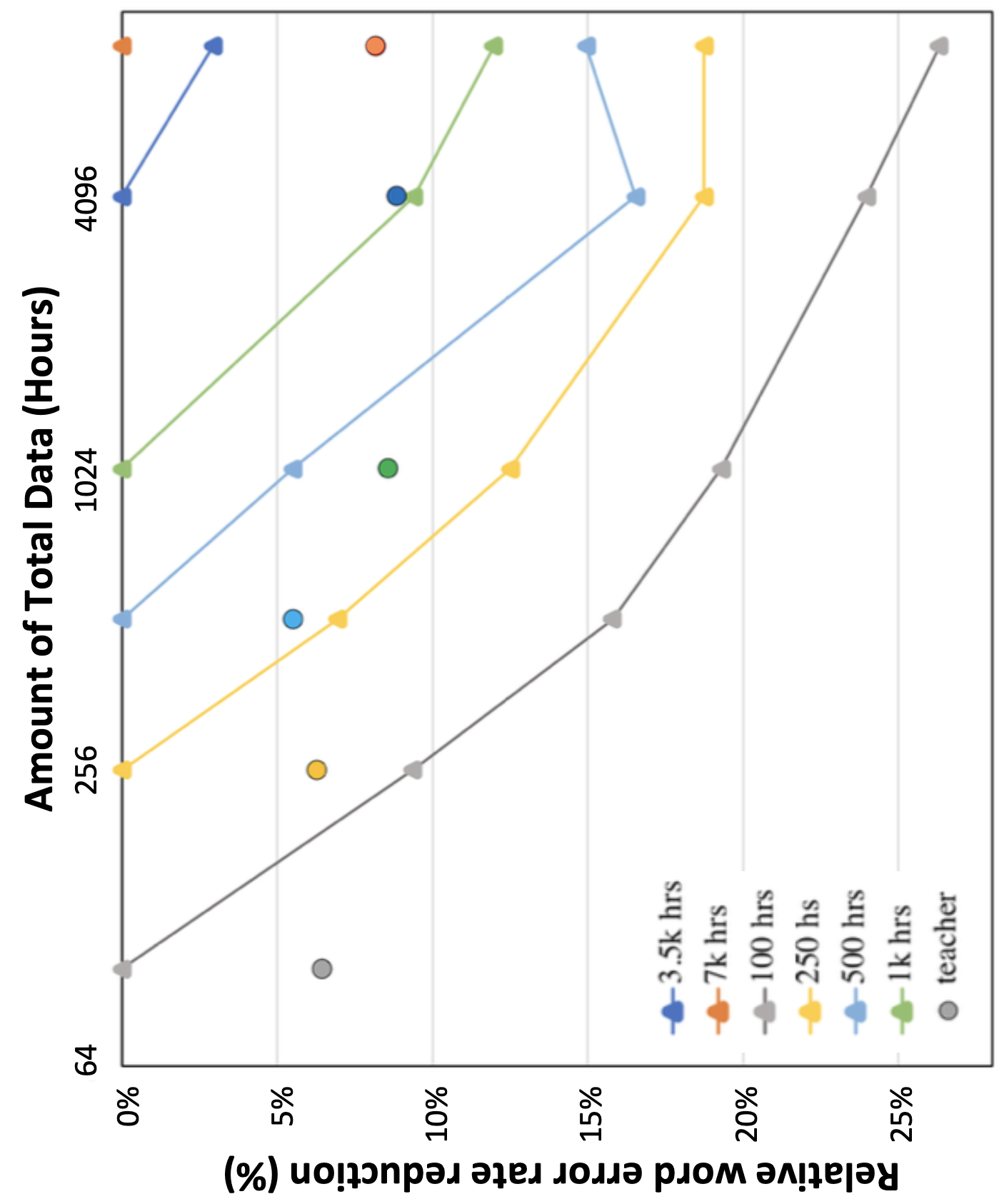}}
  \caption{\textit{On TST2: relative WERR (\%) for different amounts of unsupervised data: each curve corresponds to a fixed amount of supervised data. Markers on each curve correspond to student models trained with increasing amounts of unsupervised data. Relative WERR is computed against the model trained only on the supervised data corresponding to that curve. The circular dots correspond to the BLSTM teacher models for the curves of the same color.}}\medskip
\label{fig:unsupervised}
\end{figure}

\subsection{An Analysis of Empirical Risk with Student-Teacher Learning}
\label{sec:theory}
To simplify analysis, we restrict our setting to binary classification, and adopt notations from \cite{begin2014pac}. Let $Z = \{(x_i,y_i)\}_{i=1,\dots,N}$ be a supervised set drawn from a distribution $D$ on $\mathcal{X}\times \{-1,1\}$, where $\mathcal{X}=\{x_i\}_{i=1..N}$. Let $\mathbbm{1}[A]: A \subset \mathcal{X} \to \{0, 1\}$ be an indicator function defined as   $
    \mathbbm{1}[A](x)=\left\{
                \begin{array}{ll} 1 \quad \text{ if } x \in A \\
                  0 \quad \text{ if } x \notin A
                \end{array}
              \right.
  $
. Teacher and student network outputs are denoted $h^t(x)$ and $h^s(x)$ respectively, and are drawn from $\{-1,1\}$. The empirical teacher-student training risk is
$$R^{h^t}_{Z}(h^s):=\frac{1}{N} \sum_{i=1,\dots,N} \mathbbm{1} [h^t(x_i)\neq h^s(x_i)],$$
and the actual empirical risk of the student is
$$R_Z(h^s):=R^{y}_{Z}(h^s)=\frac{1}{N} \sum_{i=1,\dots,N} \mathbbm{1} [y_i\neq h^s(x_i)].$$

\begin{lemma}[Decomposition]
Given $h^t,h^s,y\in \{-1,1\}$, we have
\begin{equation*}
R_{Z}(h^s)=R_{Z}(h^t) \big(1-R_{acc(h^t)}^{h^t}(h^s) -R_{err(h^t)}^{h^t}(h^s)\big) + R_{acc(h^t)}^{h^t}(h^s)
\end{equation*}
where $acc(h^t) := \{x,y  \in Z| h^t(x)= y\}$ and $err(h^t) := \{x,y  \in Z| h^t(x)\neq y\} $.
\end{lemma}

\begin{proof}
Decomposing the student risk by partitioning $Z$ on whether the teacher makes a mistake
\begin{align*}
R_{Z}(h^s)=& R_{Z}(h^t) \big(1-R_{err(h^t)}^{h^t}(h^s)\big) + \big(1-R_{Z}(h^t)\big)R_{acc(h^t)}^{h^t}(h^s) \\
=&R_{Z}(h^t) \big(1-R_{acc(h^t)}^{h^t}(h^s) -R_{err(h^t)}^{h^t}(h^s)\big) +R_{acc(h^t)}^{h^t}(h^s)
\end{align*}
\end{proof}

\begin{theorem}[Truth Over Teacher] \label{thm1}
$$R_{Z}(h^s) \leq R_{Z}(h^t)   \iff  \frac{R_{err(h^t)}^{h^t}(h^s) }{ R_{acc(h^t)}^{h^t}(h^s)} \geq \frac{1}{R_{Z}(h^t)}-1$$
\end{theorem}

\begin{proof}
\begin{align*}
& R_{Z}(h^s) \leq R_Z(h^t) \\
&\iff \\
&  R_{Z}(h^t) \big(1-R_{acc(h^t)}^{h^t}(h^s) -R_{err(h^t)}^{h^t}(h^s)\big) + R_{acc(h^t)}^{h^t}(h^s) \leq R_Z(h^t) \\
&\iff \\
&R_{acc(h^t)}^{h^t}(h^s) \leq R_Z(h^t)  \big(R_{acc(h^t)}^{h^t}(h^s) +R_{err(h^t)}^{h^t}(h^s)\big) \\
&\iff \\
&\frac{1}{R_Z(h^t) } \leq 1+\frac{R_{err(h^t)}^{h^t}(h^s)}{R_{acc(h^t)}^{h^t}(h^s) } \\
&\iff \\
& \frac{R_{err(h^t)}^{h^t}(h^s) }{ R_{acc(h^t)}^{h^t}(h^s)} \geq \frac{1}{R_{Z}(h^t)}-1
\end{align*}
\end{proof}

The better the teacher is, the easier a student model fits the true label compared to a false label. In particular, if the prediction error the teacher makes $R_{Z}(h^t)$ is better than random guess, namely $R_{Z}(h^t) < 0.5$, then it is necessary that $\frac{R_{err(h^t)}^{h^t}(h^s) }{ R_{acc(h^t)}^{h^t}(h^s)} > 1$ so that $R_{Z}(h^s) \leq R_Z(h^t)$. From Theorem \ref{thm1}, for the student to be better than its teacher, the student's risk evaluated on the error set of the teacher (the erroneous teacher labels) has to be greater than the risk evaluated on the accuracy set of the teacher (the correct teacher labels) by a factor of $\frac{1}{R_z(h^t)} - 1$.

\section{Discussion}
\label{sec:discussions}
In low data regimes, our experiments surprisingly showed that a student model's performance is not upper bounded by the teacher model used in learning.  Our theoretical analysis indicates that this can happen, for example, when the student makes fewer errors on data with high teacher errors. We speculate this is because a capacity-restricted student model is able to generalize better than the teacher using the unsupervised data.

In high data regimes, we observed that the TAKD method facilitated efficient learning. We speculate that this could be due to step-wise distillation providing better calibrated posteriors for the eventual student to learn from. Indeed, it has been observed that in modern over-parameterized neural networks, posterior probabilities can become less calibrated~\cite{guo2017calibration}; using large amounts of supervised data could be making teacher models' estimates of posteriors less calibrated. 

\section{Conclusions And Future Work}
\label{sec:conclusions}
Using BLSTM and LSTM as teacher and student models respectively, we studied SSL for AMs for two tasks: (a) for small footprint on-device models; and (b) a larger footprint, but lower training data regime. Despite smaller model capacity, on-device models were able to exploit 1 million hours of unsupervised data and achieve a 14.3\% relative WER improvement after cross-entropy training and the gains could carry over to the sMBR stage. However, when we increased the supervised data from 7,000 hours to 52,000 hours, the learning efficiency decreased yielding only 7.1\% relative improvements in WER. We utilized a step-wise distillation and recovered the WER gains. For SSL in low data regimes, using knowledge distillation, we found that to achieve a performance comparable to that of a fully supervised system, the proportion of required supervised data decreased as the amount of total data increased. However, we found that in low data regimes, students can be better than teachers. In future work, we would like to extend this large scale study to distill using sequence-based discriminative criterion~\cite{Wong2016SequenceST,Manohar2018,LiZhao2020} instead of the frame-level cross-entropy criterion. Finally, it would be informative to understand if the analyses carried out in this paper using hybrid models would carry over to end-to-end ASR systems.

\section{Acknowledgements}                                                                   
\label{sec:acknowledge}
We would like to thank Minhua Wu, Jangwon Kim, Srinivas Parthasarathy, Kishore Nandury and Brian King for their helpful discussions.

\clearpage
\bibliographystyle{splncs04}
\bibliography{paper}

\end{document}